\documentclass{article}
\usepackage{authblk}
\usepackage{amssymb,amsmath,amsthm,amscd,latexsym,amsfonts}
\usepackage{mathtools}
\usepackage[T1]{fontenc}
\usepackage{graphicx,epstopdf}
\usepackage{graphicx}
\usepackage{xcolor}
\usepackage{cite}
\usepackage[a4paper,top=3cm, bottom=3cm, left=3.1cm, right=3.1cm]{geometry}

\newtheorem{thm}{Theorem}
\newtheorem{defn}{Definition}
\newtheorem{lemma}{Lemma}

\newtheorem{rk}{Remark}

\numberwithin{equation}{section} 

\newcommand{\bea}{\begin{eqnarray}}
	\newcommand{\eea}{\end{eqnarray}}

\newcommand{\Z}{\mathbb{Z}}

\def\O{\Omega}

\def\O{\Omega}

\def\L{\Lambda}

\def\Z{\mathbb{Z}}

\def \L {\Lambda}

	\title{Gradient Gibbs measures of an SOS model with alternating magnetism on Cayley trees}
	\author[1]{N.N. Ganikhodjaev}
	\author[2] {N. M. Khatamov}
	\author[3] {U.A. Rozikov}
\affil[1]{ V.I.Romanovskiy Institute of Mathematics,  9, Universitet str., 100174, Tashkent, Uzbekistan.
	
		 E-mail: nasirgani@yandex.ru} 
\affil[2]{V.I.Romanovskiy Institute of Mathematics,  9, Universitet str., 100174, Tashkent, Uzbekistan.
	
	 E-mail: nxatamov@mail.ru} 
\affil[3]{
	V.I.Romanovskiy Institute of Mathematics,  9, Universitet str., 100174, Tashkent, Uzbekistan\\
	
		Central Asian University, 264, Milliy bog St, 111221, Tashkent,  Uzbekistan\\
		
		National University of Uzbekistan,  4, Universitet str., 100174, Tashkent, Uzbekistan.
		
		 E-mail: rozikovu@yandex.ru}
	
\begin{document}
\maketitle
	\begin{abstract}
	
	The work is devoted to gradient Gibbs measures (GGMs) of a SOS model with countable set $\mathbb Z$ of spin values and having alternating magnetism on Cayley trees. This model is defined by a nearest-neighbor gradient interaction potential. Using K\"ulske-Schriever argument, based on boundary law equations, we give several $q$-height-periodic translations invariant GGMs for $q=2,3,4$. 
\end{abstract}



\section{Introduction}
It is known that macroscopic physical systems in thermodynamic equilibrium are
described by Gibbs measures, that are defined through a Hamiltonian of the system \cite{Ge}. A given Hamiltonian may have several distinct Gibbs measures; in that case, one says that there is phase coexistence. This corresponds to the first-order phase transitions, when several phases are thermodynamically stable at the same value of the relevant parameters (temperature, pressure, etc.)

For models (Hamiltonians) defined on low dimensional lattices,  it is known (see \cite{FV, FS, Ve}) non-existence of infinite-volume Gibbs measures. However, pinning a field at the origin suffices to
have a well-defined thermodynamic limit in any dimension (\cite{FV, Ve}). The resulting field
is completely characterized by its gradients. These infinite-volume random fields of gradients
are called gradient Gibbs measures, or Funaki-Spohn states \cite{FS} 
(for detailed motivations and very recent results see \cite{BiKo, BEvE}, \cite{HR} - \cite{KS}, \cite{Ra, Sh, Z1}).

This paper is devoted to GGMs of SOS model with alternating magnetism and spin values from the set of all integers  on Cayley trees.

\subsection{Definitions and known results}
Let $\Gamma^k= (V , L)$ be the uniform Cayley tree, where each vertex has $k + 1$ neighbors with $V$ being the set of vertices and $L$ the set of edges (\cite{Ro}).

For a fixed point $x^0\in V$,
$$W_n=\{x\in V\,| \, d(x,x^0)=n\},$$ 
where $d(x,y)$ is the distance between
vertices $x$ and $y$ on a Cayley tree, i.e.,
the number of edges of the shortest path connecting  $x$ and $y$.

 Denote by
$S(x)$ the set of direct successors of $x$, i.e., if $x\in W_n$, then
$$S(x)=\{y_i\in W_{n+1} \, |\,  d(x,y_i)=1, i=1,2,\ldots, k \}.$$

We consider models where spin-configuration $\omega$ is a function from the
vertices of the Cayley tree $\Gamma^k=(V, \vec L)$  to the set $\Z$ of integer numbers, where
$\vec L$ is the set of oriented edges (bonds) of the tree
(see Chapter 1 of \cite{Ro} for properties of the Cayley tree).

For any configuration  $\omega = (\omega(x))_{x \in V} \in \mathbb Z^V$ and edge $b = \langle x,y \rangle$ of $\Gamma^k$
the \textit{difference} along the edge $b$ is given by $\nabla \omega_b = \omega(y) - \omega(x)$, where $\omega_b$ is a configuration on  $b = \langle x,y \rangle$, i.e., $\omega_b=\{\omega(x), \omega(y)\}$. The gradient spin variables are defined by $\eta_{\langle x,y \rangle} = \omega(y) - \omega(x)$ for each $\langle x,y \rangle$ (see \cite{KS}, \cite{K}).

The space of \textit{gradient configurations} is denoted by $\O^\nabla$. The measurable structure on the space $\Omega^{\nabla}$ is given by $\sigma$-algebra $$\mathcal{F}^\nabla:=\sigma(\{ \eta_b \, \vert \, b \in \vec L \}).$$

Let $\mathcal{T}_{\Lambda}^{\nabla}$ be the sigma-algebra of gradient configurations outside of the finite volume $\Lambda$ is generated by all gradient variables outside of $\Lambda$ and the relative height-difference on the boundary of $\Lambda$. 

For nearest-neighboring interaction  potential $\Phi=(\Phi_b)_b$, where
$b=\langle x,y \rangle$ is an edge,  define symmetric transfer matrices $Q_b$ by
\begin{equation}\label{Qd}
	Q_b(\omega_b) = e^{- \Phi_b(\omega_b)}=\exp\left(-\beta V(\omega(x)-\omega(y))\right).
\end{equation}
Such potential $\Phi$ is called
a \textit{gradient interaction potential}.  

Following \cite{KS} and \cite{K} let us give definition of GGM.
Define the Markov (Gibbsian) specification as
$$
\gamma_\Lambda^\Phi(\sigma_\Lambda = \omega_\Lambda | \omega) = (Z_\Lambda^\Phi)(\omega)^{-1} \prod_{b \cap \Lambda \neq \emptyset} Q_b(\omega_b).
$$

\begin{defn}  The gradient Gibbs specification is defined as the family of probability kernels $\left(\gamma_{\Lambda}^{\prime}\right)_{\Lambda \Subset V}$ from $\left(\Omega^{\nabla}, \mathcal{T}_{\Lambda}^{\nabla}\right)$ to $\left(\Omega^{\nabla}, \mathcal{F}^{\nabla}\right)$ such that
	$$
	\int F(\rho) \gamma_{\Lambda}^{\prime}(d \rho \mid \zeta)=\int F(\nabla \varphi) \gamma^\Phi_{\Lambda}(d \varphi \mid \omega)
	$$
	for all bounded $\mathcal{F}^{\nabla}$-measurable functions $F$, where $\omega \in \Omega$ is any height-configuration with $\nabla \omega=\zeta$.
\end{defn}
\begin{defn} 
	A probability measure $\nu$ on $\Omega^{\nabla}$ is called a GGM if it satisfies the $DLR$ equation
	$$
	\int \nu(d \zeta) F(\zeta)=\int \nu(d \zeta) \int \gamma_{\Lambda}^{\prime}(d \tilde{\zeta} \mid \zeta) F(\tilde{\zeta})
	$$
	for every finite $\Lambda \subset V$ and for all bounded functions $F$ on $\Omega^{\nabla}$. 
\end{defn}

Now we define \emph{boundary laws} (see \cite{Z1}) which allow to describe the subset of $\mathcal{G}(\gamma)$ of all Gibbs measures.

\begin{defn}\label{def} 
	
	\begin{itemize}
		\item	A family of vectors $\{ l_{xy} \}_{\langle x,y \rangle \in \vec L}$ with $l_{xy}=\left(l_{xy}(i) : i\in \Z\right) \in (0, \infty)^\Z$ is called a {\em boundary law for the transfer operators $\{ Q_b\}_{b \in \vec L}$} if for each $\langle x,y \rangle \in \vec L$ there exists a constant  $c_{xy}>0$ such that the consistency equation
		\begin{equation}\label{eq:bl}
			l_{xy}(i) = c_{xy} \prod_{z \in \partial x \setminus \{y \}} \sum_{j \in \Z} Q_{zx}(i,j) l_{zx}(j)
		\end{equation}
		holds for every $i \in \Z$.
		\item  A boundary law $l$ is said to be {\em normalisable} if and only if
		\begin{equation}\label{Norm}
			\sum_{i \in \Z} \Big( \prod_{z \in \partial x} \sum_{j \in \Z} Q_{zx}(i,j) l_{zx}(j) \Big) < \infty
		\end{equation} at any $x \in V$.
		
		\item 	A boundary law 	is called {\em $q$-height-periodic} (or $q$-periodic) if $l_{xy} (i + q) = l_{xy}(i)$
		for every oriented edge $\langle x,y \rangle \in \vec L$ and each $i \in \Z$.
	\end{itemize}
\end{defn}

It is known that there is an one-to-one correspondence between normalisable boundary laws
and tree-indexed Markov chains \cite{Z1}.

In \cite{HKLR},  \cite{HKa}, \cite{HKb}, \cite{KS}
some non-normalisable boundary laws are used to give GGM.

For $\Lambda\subset V$, fix  a site $w \in \Lambda$.
If the boundary law $l$ is assumed to be $q$-height-periodic, then take  $s \in \mathbb{Z}_q=\{0,1,\dots,q-1\}$ and define probability measure $\nu_{w,s}$ on $\mathbb{Z}^{\{b \in \vec L \mid b \subset \Lambda\}}$ by
$$
\nu_{w,s}(\eta_{\Lambda \cup \partial \Lambda}=\zeta_{\Lambda \cup \partial \Lambda})
$$
$$=Z^\Lambda_{w,s}\prod_{y \in \partial \Lambda} l_{yy_\L}\Bigl (T_q(
s+\sum_{b\in \Gamma(w,y)}\zeta_b)
\Bigr) \prod_{b \cap \Lambda \neq \emptyset}Q_b(\zeta_b),
$$
where $Z^\Lambda_{w,s}$ is a normalization constant,  $y_\Lambda$ denotes the unique nearest neighbor of $y$ in $\Lambda$, $\Gamma(w,y)$ is the unique path from $w$ to $y$
and $T_q: \mathbb{Z} \rightarrow \mathbb{Z}_q$ denotes the coset projection.

\begin{thm} \cite{KS}
	Let $(l_{\langle xy\rangle })_{\langle x,y\rangle  \in \vec L}$ be any $q$-height-periodic boundary law for some gradient interaction potential.
	Fix any site $w \in V$ and any class label $s \in \mathbb{Z}_q$. Then
	$$	\nu_{w,s}(\eta_{\Lambda \cup \partial \Lambda}=\zeta_{\L\cup\partial\L})
	$$
	\begin{equation}
		=Z^\Lambda_{w,s} \prod_{y \in \partial \Lambda} l_{yy_\L}\Bigl (T_q(
		s+\sum_{b\in \Gamma(w,y)}\zeta_b)
		\Bigr) \prod_{b \cap \Lambda \neq \emptyset}
		Q_b(\zeta_b)
	\end{equation}
	gives a consistent family of probability measures on the gradient space $\Omega^\nabla$.
	Here $\Lambda$ with $w \in  \L \subset V$ is any finite connected set,
	$\zeta_{\L\cup\partial\L} \in \Z^{\{b \in \vec L \mid b \subset (\L \cup \partial\L)\}}$ and $Z^\Lambda_{w,s}$ is a normalization constant.
\end{thm}
The measure $\nu_{w,s}$ is called a pinned gradient measure.

If $q$-height-periodic boundary law  and the underlying potential are translation invariant then it is possible to obtain
probability measure $\nu$ on the gradient space by mixing the pinned gradient measures:

\begin{thm}\cite{KS}\label{KF}	
	Let a $q$-height-periodic boundary law $l$  and  its gradient interaction potential be translation invariant.
	Let $\Lambda \subset V$ be any finite connected set and let $w\in \Lambda$ be any vertex. Then the measure $\nu $ with marginals given by
	\begin{equation}
		\nu (\eta_{\L\cup\partial \L} = \zeta_{\L\cup\partial\L}) = Z_\L \ \left(\sum_{s\in\Z_q}  \prod_{y \in \partial\L} l \big(s + \sum_{b \in \Gamma(w,y)} \zeta_{b}\big)  \right)\prod_{b \cap \L \neq \emptyset} Q(\zeta_b),
	\end{equation}
	where $Z_\L$ is a normalisation constant, defines a translation invariant GGM on $\Omega^\nabla$.
\end{thm}

The aim of this paper is to study GGMs of gradient potential constructed by an SOS model with alternating magnetism. By Theorem \ref{KF} each $q$-height-periodic boundary law $l$, which is independent on edges of tree, defines a translation invariant GGM. We find  $q$-height-periodic boundary laws for this SOS model and use Theorem \ref{KF} to construct GGMs corresponding to these laws. 

\subsection{The boundary law equation for the SOS model.}

In this subsection for  $\sigma:x\in V\mapsto \sigma(x)\in \mathbb Z$, consider Hamiltonian
of SOS model, with alternating magnetism, i.e.,
\begin{equation}\label{f1}
	H(\sigma)=-J\sum_{\langle x, y\rangle, x, y\in V; }\alpha\left(\mid \sigma(x)-\sigma(y)\mid\right)\mid \sigma(x)-\sigma(y)\mid,
\end{equation}
where  $J\in \mathbb R$ and
$$\alpha(\mid m\mid)=%
\begin{cases} 1, \ \ \mbox{if} \ \ m\in 2\mathbb{Z} \\
	-1, \ \ \mbox{if} \ \ m\in 2\mathbb{Z}+1.
\end{cases}
$$
We consider the set $\mathbb{Z}$ as the set of vertices of a graph $G$.
We use the graph $G$  to define a $G$-admissible configuration as follows.
A configuration $\sigma$ is called a
$G$-\textit{admissible configuration} on the Cayley tree, if $\{\sigma (x), \, \sigma (y)\}$ is the edge of the graph $G$
for any pair of nearest neighbors $x,y$ in $V$. We
let $\Omega^G$ denote the set of $G$-admissible configurations.

Let $L(G)$ be the set of edges of a graph $G$. We let $A\equiv A^G=\big(a_{ij}\big)_{i,j\in \mathbb Z}$ denote the adjacency
matrix of the graph $G$, i.e.,

$$a_{ij}=a_{ij}^G=%
\begin{cases} 1 \ \ \mbox{if} \ \ \{i,j\}\in L(G), \\
	0 \ \ \mbox{if} \ \ \{i,j\}\notin L(G).
\end{cases}
$$

For any $y\in V$ and $x\in S(y)$ (so that
$d(x^0, y)=d(x^0, x)-1$), for convenience of notation we set
\begin{equation}\label{zi}
z_{i,x}:=l_{xy}(i),\qquad i\in\mathbb Z.
\end{equation}

Assume $z_{0,x}=l_{xy}(0)\equiv 1$, then the boundary law equation  (\ref{eq:bl}) corresponding to our model on the $G$-admissible configurations, by the new notation (\ref{zi}) becomes (cf. \cite{BR})
\begin{equation}\label{f7}
\begin{array}{ll}
z_{2i,x}=\prod\limits_{y \in S(x)} \frac{a_{2i,0}\theta^{2\mid i\mid}+\sum\limits_{j\in \mathbb{Z}_0} {a_{2i,2j}\theta^{2\mid i-j\mid}z_{2j,y}+\sum\limits_{j\in \mathbb{Z}} {a_{2i,2j+1}\theta^{-\mid 2(i-j)-1\mid}z_{2j+1,y}}}}{a_{0,0}+\sum\limits_{j\in \mathbb{Z}_0} {a_{0,2j}\theta^{2\mid j\mid}z_{2j,y}}+\sum\limits_{j\in \mathbb{Z}} {a_{0,2j+1}\theta^{-\mid 2j+1\mid}z_{2j+1,y}}}, \\[6mm]
z_{2i+1,x}=\prod\limits_{y \in S(x)} \frac{a_{2i+1,0}\theta^{-\mid 2i+1\mid}+\sum\limits_{j\in \mathbb{Z}_0} {a_{2i+1,2j}\theta^{-\mid 2(i-j)+1\mid}z_{2j,y}+\sum\limits_{j\in \mathbb{Z}} {a_{2i+1,2j+1}\theta^{2\mid i-j\mid}z_{2j+1,y}}}}{a_{0,0}+\sum\limits_{j\in \mathbb{Z}_0} {a_{0,2j}\theta^{2\mid j\mid}z_{2j,y}}+\sum\limits_{j\in \mathbb{Z}} {a_{0,2j+1}\theta^{-\mid 2j+1\mid}z_{2j+1,y}}}, 
\end{array}
\end{equation}
where $i\in \mathbb{Z}$, $\theta=\exp(-J\beta)$.

In general, it is very difficult to solve the system (\ref{f7}). The difficulty depends on the structure of graph $G$.
In this paper we choose the graph $G$, defined by (see Fig.\ref{fi})
$$a_{ij}=%
\begin{cases} 1, \ \ \mbox{if} \ \ \ i=j \ \ \mbox{or}\ \ \mid i-j\mid=1, \ \ i,j\in\mathbb{Z}, \\
0, \ \ \mbox{otherwise}.
\end{cases}
$$

\begin{figure}[h]
\begin{center}
   \includegraphics[width=13cm]{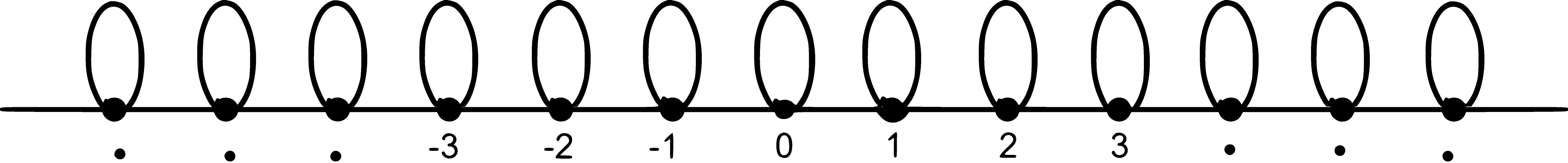}
\end{center}
     \caption{The graph $G$ with the set $\mathbb Z$ of vertices.} \label{fi}
\end{figure}

For this concrete graph from (\ref{f7}) we get
\begin{equation}\label{f8}
z_{i,x}=\prod_{y \in S(x)} \frac{\theta z_{i,y}+z_{i-1,y}+z_{i+1,y}}{\theta+z_{-1,y}+z_{1,y}}, \ \
i\in \mathbb{Z}_0 
\end{equation}

\section{Constant solutions: Translation invariant measures}
In this section we find some constant solutions of (\ref{f8}), i.e., $z_x=z\in \mathbb R_{+}^{\infty}$ does not depend on vertices of the Cayley tree. In this case from (\ref{f8}) we get 
\begin{equation}\label{f9}
	z_{i}=\left(\frac{\theta z_{i}+z_{i-1}+z_{i+1}}{\theta+z_{-1}+z_{1}}\right)^k, \ \
	i\in \mathbb{Z}_0 
\end{equation}
here $\theta>0, \ z_i>0$, $z_0=1$.

\subsection{2-periodic solution of (\ref{f9})}
In (\ref{f9}) we assume 
$$
z_j= \begin{cases}
1, \ \ \mbox{if}  \ \ j \ \ \mbox{is even}, \\
y, \ \ \mbox{if} \ \ j \ \ \mbox{is odd}.
\end{cases}
$$
Then from equation (\ref{f9}) we get

\begin{equation}\label{f11}
y=\left(\frac{\theta y+2}{\theta+2y}\right)^{k}.
\end{equation}
Denoting $a=\sqrt[k]{y}$ from (\ref{f11}) we obtain 
\begin{equation}\label{ab3}
	2a^{k+1}-\theta a^k+\theta a-2=0
\end{equation}

The equation (\ref{ab3}) has the solution $a=1$ independently of the parameters $(\theta, k)$. Dividing both sides of (\ref{ab3}) by $a-1$
we get
\begin{equation}\label{uy22}
	2a^k+(2-\theta)(a^{k-1}+a^{k-2}+\dots+a)+2=0.
\end{equation}
The following lemma gives the number of solutions to equation (\ref{uy22}):
\begin{lemma}\label{l6}\cite[Lemma 1]{HR} For each $k\geq 2$,
	there is exactly one critical value of $\theta$, i.e., $\theta_{\mathrm{c}}=\theta_{\mathrm{c}}(k):={2(k+1)\over k-1}$, such that
	\begin{enumerate}
		\item if $\theta<\theta_{\mathrm{c}}$ then (\ref{uy22}) has no positive solution;
		\item if $\theta=\theta_{\mathrm{c}}$ then the equation has a unique solution $a=1$;
		\item if $\theta>\theta_{\mathrm{c}}$,	then it has exactly two solutions (both different from 1) denoted as $a_1$, $a_2$;
	\end{enumerate}
\end{lemma}

Thus the corresponding solutions of (\ref{f11})  are 
\begin{equation}\label{bir}
\begin{array}{lll}
	1) \ \ 1 \ \mbox{for} \ \ \theta\leq \theta_{c};\\
	2) \ \ 1, \ a_1^k, \ a_2^k, \ \ \mbox{for} \ \ \theta>\theta_{c}.
\end{array}
\end{equation}
				
By Theorem \ref{KF} each $2$-height-periodic boundary law  defines a translation invariant GGM. Therefore, solutions (\ref{bir}) allow us to formulate the following  result.

\begin{thm} Let $k\geq2$,  $\theta_{c}=\frac {2(k+1)}{k-1}$. Then for the Hamiltonian (\ref{f1}) on the $G$-admissible configurations space with graph $G$ given in Fig. \ref{fi}, the number $\nu_2(k,\theta)$ of 2-height-periodic GGMs is given by the following formula
	$$\nu_2(k,\theta)=\left\{\begin{array}{lll}
		1, \ \ \mbox{if} \ \ \theta\leq \theta_c\\[2mm]
		3, \ \ \mbox{if} \ \ \theta>\theta_c.
		\end{array}\right.
 $$
\end{thm}

\subsection{3-periodic solutions of (\ref{f9})}

In this subsection we are interested to find 3-periodic solutions of the system (\ref{f9}) which have the form
\begin{equation}\label{zz}
z_j= \begin{cases}
1, \ \ \mbox{if} \ \ j\equiv0 \mod3, \\
x, \ \ \mbox{if} \ \ j\equiv1 \mod3, \\
y, \ \ \mbox{if} \ \ j\equiv2 \mod3.
\end{cases}
\end{equation}
Then the system becomes:
\begin{equation}\label{f28}
\left\{
\begin{array}{ll}
x=\left(\frac{1+y+\theta x}{\theta+y+x}\right)^{k}, \\[4mm]
y=\left(\frac{1+x+\theta y}{\theta+y+x}\right)^{k}.
\end{array}
\right.\end{equation}
Surprisingly this system of equations coincides with the system (3.2) in \cite{KRK} which describes translation-invariant Gibbs measures of the $q$-state Potts model. Our system (\ref{f28}) is a particular case which corresponds to $q=3$. 

By Theorem 1 of \cite{KRK} for $q=3$ we have 

\begin{lemma}\label{an}  Let $k\geq2$, $\theta_{\rm cr}={k+2\over k-1}$. Then for the system (\ref{f28}) there is  $\theta=\theta_{c,1}$ such that 
\begin{itemize}	
\item[1.] If	$\theta<\theta_{c,1}$ then there is unique solution: $(1,1)$;
\item[2.] If $\theta=\theta_{c,1}$ or $\theta=\theta_{\rm cr}$ then there are 4 solutions: 
$$(1,1), \, (x_{1}^{(3)},x_{1}^{(3)}), \, (1,x_{3}^{(3)}), \, (x_{3}^{(3)},1);$$
\item[3.] If $\theta>\theta_{c,1}$, $\theta\ne \theta_{\rm cr}$
then there are 7 solutions 
$$(1,1), \, (x_{1}^{(3)},x_{1}^{(3)}), \, (x_{2}^{(3)},x_{2}^{(3)}), \, (1,x_{3}^{(3)}), \, (x_{3}^{(3)},1), \, (1,x_{4}^{(3)}), \, (x_{4}^{(3)},1).$$
\end{itemize}
\end{lemma}
\begin{rk}\label{rk}
	Note that (see \cite{HKLR} and \cite{HKb})	if a height-periodic boundary law is obtained from another one by a cyclic
	shift, then it leads to the same GGM.  
\end{rk}
Now by Lemma \ref{an}, Remark \ref{rk} and Theorem \ref{KF} we get the following result. 
\begin{thm} For the model (\ref{f1}) with graph $G$ given in Fig. \ref{fi}, the number $\nu_3(k,\theta)$ of 3-height-periodic (of the form (\ref{zz})) GGMs is given by the following formula
	$$\nu_3(k,\theta)=\left\{\begin{array}{lll}
		1, \ \ \mbox{if} \ \ \theta<\theta_{c,1}\\[2mm]
		3, \ \ \mbox{if} \ \ \theta=\theta_{c,1}, \theta_{\rm cr}\\[2mm]
		5, \ \ \mbox{if} \ \ \theta\in (\theta_{c,1}, +\infty)\setminus \{\theta_{\rm cr}\}.
	\end{array}\right.\quad
	$$
\end{thm}

\subsection{4-periodic solutions of (\ref{f9})}

In the system (\ref{f9}) we are going to find solutions of the form
\begin{equation}\label{zzz}
z_j= \begin{cases}
1, \ \ \mbox{if} \ \ j\equiv0 \mod4, \\
x, \ \ \mbox{if} \ \ j\equiv1 \mod4, \\
z, \ \ \mbox{if} \ \ j\equiv2 \mod4,\\
y, \ \ \mbox{if} \ \ j\equiv3 \mod4.
\end{cases}\quad
\end{equation}
i.e. we are considering the following sequence 
\begin{equation}\label{xyz}
	 \dots, 1, x, z, y, 1, x, z, y, 1, x, z, y, 1, \dots\quad
\end{equation}	 
Then $x, z, y$ should satisfy
\begin{equation}\label{fa}
	\left\{
	\begin{array}{lll}
		x=\left(\frac{\theta x+1+z}{\theta+x+y}\right)^{k}, \\[3mm]
		z=\left(\frac{x+y+\theta z}{\theta+x+y}\right)^{k}, \\[3mm]
		y=\left(\frac{1+\theta y+z}{\theta+x+y}\right)^{k}.
	\end{array}
	\right.\end{equation}

{\bf CASE $z=1$}. Note that $z=1$ is a solution to the second equation of (\ref{fa}), independently on values of parameters $k\geq 1$, $\theta>0$ and unknowns $x,y$. In this case the system becomes 
\begin{equation}\label{f38}
\left\{
\begin{array}{ll}
x=\left(\frac{\theta x+2}{\theta+x+y}\right)^{k}, \\[3mm]
y=\left(\frac{\theta y+2}{\theta+x+y}\right)^{k}.
\end{array}
\right.\end{equation}

Subtracting from the first equation of system (\ref{f38}) the second one, we get
\begin{equation}\label{f39}
(x-y)\left[1-\frac{\theta\left((\theta x+2)^{k-1}+...+(\theta y+2)^{k-1}\right)}{(\theta+x+y)^{k}}\right]=0.
\end{equation}

Consequently, $x=y$ or
\begin{equation}\label{f40}
(\theta+x+y)^{k}=\theta\left((\theta x+2)^{k-1}+...+(\theta y+2)^{k-1}\right).
\end{equation}

{\bf SubCASE:} $x=y.$ In this case our 4-periodic solutions coincide with 2-periodic ones. 
%
%
%

{\bf SubCASE:} $x\ne y$. In this case for simplicity we take $k=2$. 
Denote
\begin{equation}\label{the}\theta_{c}^{(3)}=\frac{2}{3}\sqrt[3]{54+6\sqrt{33}}+\frac{8}{\sqrt[3]{54+6\sqrt{33}}}+2\approx6.766.
	\end{equation}
\begin{lemma}\label{4b}  Let $k=2.$ Then for solutions of the system (\ref{f38}), with $x\ne y$, the following hold 
\begin{itemize}
  \item if $\theta \leq 2$  then there is no solution (satisfying $x\ne y$);
  \item if $2<\theta\leq \theta_{c}^{(3)}$  then there are exactly two such solutions:  $$\left(x_{3}^{(4)},y_{3}^{(4)}\right), \, \left(y_{3}^{(4)},x_{3}^{(4)}\right);$$
   \item if $\theta>\theta_{c}^{(3)}$ then there are exactly four solutions (satisfying $x\ne y$):  $$\left(x_{3}^{(4)},y_{3}^{(4)}\right), \, \left(y_{3}^{(4)},x_{3}^{(4)}\right), \, \left(x_{4}^{(4)},y_{4}^{(4)}\right), \, \left(y_{4}^{(4)},x_{4}^{(4)}\right).$$
\end{itemize}
\end{lemma}
\begin{proof} By (\ref{f39}), for  $k=2$ and $x\neq y$ we get  $$(\theta+x+y)^2=\theta(\theta(x+y)+4).$$
Solving this with respect to $x+y$ we get 
\begin{equation}\label{f42}
(x+y)_{1,2}=\frac{\theta^{2}-2\theta \pm\sqrt{D}}{2}=\phi_{1,2}(\theta),
\end{equation}
where   $D=\theta(\theta^{3}-4\theta^{2}+16).$

Note that $D>0$, for any $\theta >0$. Indeed, 

 $D>0$ iff $h(\theta):=\theta^{3}-4\theta^{2}+16>0.$ It is easy to see that 
 $$\min_{\theta>0}h(\theta)=h\left(\frac{8}{3}\right)=\frac{176}{27}>0.$$ 

One can check that $\phi_{1}(\theta)>0$ for any $\theta >0$ and  $\phi_{2}(\theta)>0$ for any $\theta >4.$

{\bf Subcase:} $x+y=\phi_1(\theta)$ for $\theta>0$. Substituting this expression into the systems of equations (\ref{f38}), for $k=2$, we obtain a quadratic equation with respect to $x$, i.e.
\begin{equation}\label{f43}
x^{2}-\phi_{1}(\theta)x+\frac{4}{\theta^{2}}=0.
\end{equation}
For discriminant of this equation we have 
$$D_{1}=\phi_{1}^{2}(\theta)-\frac{16}{\theta^{2}}.$$
$D_{1}\geq 0$ is equivalent to  $g_{1}(\theta)=\phi_{1}(\theta)-\frac{4}{\theta}\geq 0$. That is 
\begin{equation}\label{f44}
\sqrt{\theta(\theta^{3}-4\theta^{2}+16)}\geq \frac{8-\theta^{3}+2\theta^{2}}{\theta}.
\end{equation}

Let $8-\theta^{3}+2\theta^{2}\geq0.$ Then from (\ref{f44}) we get $(\theta-2)(\theta-\theta_{c}^{(3)})<0$.
The inequality $8-\theta^{3}+2\theta^{2}\geq0$ holds iff  $\theta-\theta_{\ast}^{(2)}\leq0,$ where $$\theta_{\ast}^{(2)}=\frac{1}{3}\sqrt[3]{116+12\sqrt{93}}+\frac{4}{\sqrt[3]{116+12\sqrt{93}}}+2\approx2.931.$$

Thus in the case  $8-\theta^{3}+2\theta^{2}\geq0$ the solution of (\ref{f44}) is $\theta\in\left[2,\theta_{\ast}^{(2)}\right].$

Let $8-\theta^{3}+2\theta^{2}\leq0.$ This is equivalent to  $\theta-\theta_{\ast}^{(2)}\geq0.$ Then one can see that in this case the solution of (\ref{f44}) is $\theta\in\left[\theta_{\ast}^{(2)};+\infty\right).$

Thus, the solution to the inequality (\ref{f44}) is $\theta\geq 2$, i.e. $D_{1}\geq 0$ for $\theta\geq 2$. So the equation (\ref{f43}) has two (resp. one) positive solutions for $\theta>2$ (resp. $\theta=2$):

\begin{equation}\label{f45}
x_{3}^{(4,1)}={1\over
2}\phi_{1}(\theta)+\frac{\sqrt{D_{1}}}{2},
\\ \  \  \ x_{3}^{(4,2)}={1\over
2}\phi_{1}(\theta)-\frac{\sqrt{D_{1}}}{2}.
\end{equation}

Since $x+y=\phi_1(\theta),$ we have  $$x_{3}^{(4,1)}=y_{3}^{(4,2)}, \ \ y_{3}^{(4,1)}=x_{3}^{(4,2)},$$ i.e., the solutions of the system are symmetrical: $$\left(x_{3}^{(4)},y_{3}^{(4)}\right), \ \  \left(y_{3}^{(4)},x_{3}^{(4)}\right).$$

{\bf Subcase:} $x+y=\phi_{2}(\theta)$ for $\theta>4$. In this case from (\ref{f38}), we get 
\begin{equation}\label{f46}
x^{2}-\phi_{2}(\theta)x+\frac{4}{\theta^{2}}=0.
\end{equation}
By a similar analysis as in the previous case one can prove that 
the equation (\ref{f46}) has two positive solutions iff  $\theta>\theta_{c}^{(3)}$:

\begin{equation}\label{f48}
x_{4}^{(4,1)}={1\over
2}\phi_{2}(\theta)+\frac{\sqrt{D_{2}}}{2},
\\ \  \  \ x_{4}^{(4,2)}={1\over
2}\phi_{2}(\theta)-\frac{\sqrt{D_{2}}}{2}.
\end{equation}

By symmetry these solutions have the form  $$\left(x_{4}^{(4)},y_{4}^{(4)}\right), \ \ \left(y_{4}^{(4)},x_{4}^{(4)}\right). $$
\end{proof}

These solutions generate only two distinct 4-periodic (different from 2- and 3- periodic ones) sequences:

\begin{equation}\label{too}
\theta\geq \theta_{c}^{(3)}: \ \	\begin{array}{ll}
		\dots, 1, x_{3}^{(4)}, 1, y_{3}^{(4)}, 1, x_{3}^{(4)}, 1, y_{3}^{(4)}, 1, \dots\\[2mm]
		\dots, 1, x_{4}^{(4)}, 1, y_{4}^{(4)}, 1, x_{4}^{(4)}, 1, y_{4}^{(4)}, 1, \dots
	\end{array}
\end{equation}
{\bf CASE} $z\ne 1$. Denoting $u=\sqrt[k]{x}$, $v=\sqrt[k]{z}$, $w=\sqrt[k]{y}$ by (\ref{fa}) we get 

\begin{equation}\label{far}
	\left\{
	\begin{array}{lll}
		u=\frac{\theta u^k+v^k+1}{\theta+u^k+w^k}, \\[3mm]
		v=\frac{u^k+\theta v^k+w^k}{\theta+u^k+w^k}, \\[3mm]
		w=\frac{1+v^k+\theta w^k}{\theta+u^k+w^k}.
	\end{array}
	\right.\end{equation}

Consider the case $k=2$ then from the second equation of 
(\ref{far}) we get (since $v\ne 1$)
\begin{equation}\label{v}v-1={\theta(v^2-1)\over \theta +u^2+w^2} \ \ \Leftrightarrow \ \
 \theta +u^2+w^2=\theta (1+v).\end{equation}
 From the first and last equations of (\ref{far}) we obtain
 \begin{equation}\label{uw}
 	u-w={\theta(u^2-w^2)\over \theta +u^2+w^2} \ \ \Leftrightarrow \ \ u=w \ \ \mbox{or} \ \
 \theta +u^2+w^2=\theta (u+w).
 \end{equation}
\begin{equation}\label{uu}
	u+w={\theta(u^2+w^2)+2(1+v^2)\over \theta +u^2+w^2}.
\end{equation}

{\bf Case $u=w$}. In this case the system is in the form
\begin{equation}\label{fu}
	\left\{
	\begin{array}{lll}
		u=\frac{\theta u^2+v^2+1}{\theta+2u^2}, \\[3mm]
		v=\frac{2u^2+\theta v^2}{\theta+2u^2}.
		\end{array}
	\right.\end{equation}
By (\ref{v}) for $u=w$ from (\ref{fu}) we get 
$$v={2\over \theta}\left({2+\theta^2 v+2v^2\over 2\theta(1+v)}\right)^2.$$
That is 
\begin{equation}\label{q4}
4v^4-2\theta^2(\theta-2)v^3+(\theta^4-4\theta^3+8)v^2-2\theta^2(\theta-2)v+4=0.
\end{equation}
Introduce 
$$\xi=v+{1\over v}>2.$$ Then from (\ref{q4}) we get 
\begin{equation}\label{x4}
	4\xi^2-2\theta^2(\theta-2)\xi+\theta^3(\theta-4)=0.
\end{equation}
Solutions of which are 
$$\xi_{1,2}={\theta\over 4}\left(\theta^2-2\theta\mp\sqrt{\theta(\theta^3-4\theta^2+16)}\right).$$
We should have $\xi_i>2$. Recall $\theta_c^{(3)}$ defined by (\ref{the}).
A numerical analysis shows that 
$$\xi_1>2, \ \ \mbox{iff} \ \ \theta>\theta_{c}^{(3)},$$
$$\xi_2>2, \ \ \mbox{iff} \ \ \theta>2.$$
If $\xi_i>2$ then we denote by $\hat v_i$ and $1/\hat v_i$ two solutions of $v+1/v=\xi_i$. Then corresponding $\hat u_i$ can be found from equality 
$2\hat u_i^2=\theta v_i$, i.e., denote $x_i=\hat u_i^2={\theta\over 2}\hat v_i$ and $\hat x_i={\theta\over 2\hat v_i}$. 
Thus we obtain the following 4-periodic sequences
\begin{equation}\label{41}
\theta>\theta_{c}^{(3)}: \begin{array}{ll}
	\dots, 1, x_1, \hat v^2_1, x_1, 1, x_1, \hat v^2_1, x_1, 1, \dots\\[2mm]
		\dots, 1, \hat x_1, \hat v^{-2}_1, \hat x_1, 1, \hat x_1, \hat v^{-2}_1, \hat x_1, 1, \dots\\[2mm]
	\end{array}
\end{equation}
\begin{equation}\label{42}
\theta>2: \begin{array}{ll}
	\dots, 1, x_2, \hat v^2_2, x_2, 1, x_2, \hat v^2_2, x_2, 1, \dots\\[2mm]
	\dots, 1, \hat x_2, \hat v^{-2}_2, \hat x_2, 1, \hat x_2, \hat v^{-2}_2, \hat x_2, 1, \dots\\[2mm]
\end{array}
\end{equation}

{\bf Case $u\ne w$.}
Using (\ref{v}) and (\ref{uw}) (for $u\ne w$) from (\ref{uu}) we get
\begin{equation}\label{ua}
\begin{array}{ll}
	u+w=1+v\\[2mm]
		u+w={\theta^2v+2(1+v^2)\over \theta(1+v)}.
		\end{array}
\end{equation}
Thus 
\begin{equation}\label{ub}
	1+v={\theta^2v+2(1+v^2)\over \theta(1+v)} \ \ \Leftrightarrow \ \ (\theta-2)(v^2-\theta v+1)=0.
	\end{equation}
The last equation has infinitely many solutions (i.e. $\forall v>0$) if $\theta=2$ and two solutions 
\begin{equation}\label{v12}v_{1,2}={1\over 2}(\theta\pm \sqrt{\theta^2-4}), \ \ \mbox{if} \ \ \theta\geq 2.
	\end{equation}

If $\theta=2$ then for each $v>0$ using (\ref{v}) from the first equation of (\ref{far}) (with $k=2$) we get 
$$2u^2-2(1+v)u+v^2+1=0 \ \ \Leftrightarrow \ \ (u-1)^2+(u-v)^2=0.$$
The last equation has unique solution $u=v=1$. For these values we get $w=1$ too.

If $\theta> 2$ for a given solution $v$ to find corresponding $u$ and $w$ we use (\ref{v}) and from the first equation of (\ref{far}) we get
$$\theta u^2-\theta (1+v)u+v^2+1=0.$$ 
Since $v^2+1=\theta v$ we have 
$$\theta u^2-\theta (1+v)u+\theta v=0 \ \ \Leftrightarrow \ \  u^2-(1+v)u+ v=0.$$
This has solutions  $u=1$ and $u=v$. Similarly, from the third equation of (\ref{far}) we get $w=1$ and $w=v$.
Hence solutions $(u,v,w)$ under conditions $u\ne w$, $v\ne 1$ have the form 
\begin{equation}\label{v1}(1, v_1, v_1), (v_1, v_1, 1), 
(1, v_2, v_2), (v_2, v_2, 1).
\end{equation}
Going back to (\ref{xyz}) we note that solutions (\ref{v1}), with $v_i$ defined in (\ref{v12}), generate only two distinct sequences
\begin{equation}\label{to}
 \theta\geq 2:\ \	\begin{array}{ll}
	\dots, 1, 1, v_1^2, v_1^2, 1, 1, v_1^2, v_1^2, 1, 1, v_1^2, v_1^2, \dots\\[2mm]
	\dots, 1, 1, v_2^2, v_2^2, 1, 1, v_2^2, v_2^2, 1, 1, v_2^2, v_2^2, \dots
\end{array}
\end{equation}
By Lemma \ref{4b}, Remark \ref{rk}, above obtained solutions and Theorem \ref{KF} we get the following result. 
\begin{thm} For the model (\ref{f1}) with graph $G$ given in Fig. \ref{fi}, the number $\nu_4(k,\theta)$ of 4-height-periodic GGMs (which are NOT 2- and 3-periodic) is given by the following formula 
	$$\nu_4(2,\theta)=\left\{\begin{array}{lllll}
		0, \ \ \mbox{if} \ \ \theta\leq 2\\[2mm]
		5, \ \ \mbox{if} \ \ \theta\in (2,\theta_{c}^{(3)})\\[2mm]
		6,  \ \ \mbox{if} \ \ \theta=\theta_{c}^{(3)}\\
		8,  \ \ \mbox{if} \ \ \theta>\theta_{c}^{(3)}
	\end{array}\right.
	$$
\end{thm}
\begin{rk}
	Some our results can be generalized for $k\geq 3$. For example, if $u=1$ and $v=w$ then system (\ref{far}) is reduced to $v^{k+1}-(\theta+1)v^k+(\theta+1)v-1=0$. This equation is very similar to (\ref{ab3}).
\end{rk}
\section*{Acknowledgements}
The work supported by the fundamental project (number: F-FA-2021-425)  of The Ministry of Innovative Development of the Republic of Uzbekistan.
Rozikov thanks Institut des Hautes \'Etudes Scientifiques (IHES), Bures-sur-Yvette, France and the IMU-CDC for support of his visit to IHES. We thank the referees for carefully reading manuscript and useful comments.

\section*{Statements and Declarations}
	
{\bf	Conflict of interest statement:} 
On behalf of all authors, the corresponding author (U.A.Rozikov) states that there is no conflict of interest.

\section*{Data availability statements}
The datasets generated during and/or analysed during the current study are available from the corresponding author on reasonable request.

\end{document}